\documentclass[12pt]{article}
\usepackage[utf8]{inputenc}
\usepackage{graphicx}
\usepackage{amsmath,amsfonts,bm,amsthm,amssymb}
\usepackage{multirow}
\usepackage[margin=1in]{geometry}
\usepackage{mathtools}  

\usepackage{comment}

\usepackage{xcolor}

\usepackage{float}
\usepackage{caption}
\usepackage{subcaption}
\usepackage{url}

\newtheorem{theorem}{Theorem}

\theoremstyle{definition}

\DeclareMathOperator{\E}{\mathbb{E}}

\let\originalleft\left
\let\originalright\right
\renewcommand{\left}{\mathopen{}\mathclose\bgroup\originalleft}
\renewcommand{\right}{\aftergroup\egroup\originalright}

\usepackage{natbib}

\usepackage{booktabs,amsmath,caption,array}
\newcolumntype{C}[1]{>{\centering\arraybackslash}m{#1}}

\usepackage{newtxtext,newtxmath}
\usepackage{accents}
\newlength{\dhatheight}

\title{\bf Exact-corrected confidence interval for risk difference in noninferiority binomial trials} 
  \author{Nour Hawila  \& Arthur Berg\footnote{Corresponding author: berg@psu.edu}\\
    Division of Biostatistics \& Bioinformatics, Penn State University}

\date{}

\begin{document}

\maketitle

\begin{abstract}
A novel confidence interval estimator is proposed for the risk difference in noninferiority binomial trials. The confidence interval is consistent with an exact unconditional test that preserves the type-I error, and has improved power, particularly for smaller sample sizes, compared to the confidence interval by \cite{Chan:99}.  The improved performance of the proposed confidence interval is theoretically justified and demonstrated with simulations and examples. An R package is also distributed that implements the proposed methods along with other confidence interval estimators. 

\end{abstract}

\section{Introduction}

We consider a noninferiority trial with binary outcome and risk difference as the treatment effect. The noninferiority trial design incorporates a noninferioiry margin, $\delta_0$, and generalizes the standard comparative binomial trial corresponding to $\delta_0=0$.  In \cite{Chan:98}, a class of exact-based tests are described for noninferiority binomial trials with type-I error rates that are guaranteed to be bounded by the level of the test. These exact-based procedures do not leverage the conditional distribution of a sufficient statistic, like that of the Fisher's exact test, but rather produces an exact unconditional test using a maximization method \citep{Boschloo:70,McDonald:77,Lehmann:06,Basu:11}. For standard comparative trials ($\delta_0=0$), such exact unconditional tests were shown to be more powerful than Fisher's conditional exact test \citep{Haber:86,Suissa:85}. 

As described in \cite{Wasserstein:19}, it is often not appropriate to just report p-value results; interval estimates of the effect size should also be reported. The unconditional exact method of \cite{Chan:98} does not immediately yield a corresponding confidence interval estimator, but \cite{Chan:99} produce a confidence interval estimator that does leverage the exact unconditional test. However, we show this confidence interval corresponds to a statistical test that is more conservative than the  exact unconditional test of \cite{Chan:98}.  Less conservative confidence interval estimators have been proposed in \cite{Miettinen:85} and \cite{Farrington:90}, but these confidence interval estimators are based on asymptotic distributions and correspond to statistical tests that do note necessarily preserve type-I error rates. 

We introduce a novel confidence interval estimator -- called the exact-corrected estimator -- that is less conservative and more powerful than the Chan \& Zhang interval, but that also corresponds to a statistical test with preserved type-I error. The approach modifies the pivotal quantity used to produce the asymptotic confidence interval in \cite{Miettinen:85}, referred to as the $\delta$-projected Z-score, and tacks on a correction factor to produce a confidence interval that is consistent with the exact test of \cite{Chan:98}. The proposed exact-corrected interval estimator is particularly novel in that it explicitly incorporates the noninferiority margin in the estimator, so different pre-determined noninferiority margins will result in different confidence intervals. 

Next we precisely define the statistical hypotheses being considered for a noninferiority binomial trial and formalize the statistical modeling framework. Then we discuss the implementation of Chan's exact p-value followed by introducing the $\delta-$projected Z-score as the choice of test statistic for building an asymptotic confidence interval, which also serves as the basis of the Chan \& Zhang confidence interval. We then introduce the proposed exact-corrected $\delta-$projected confidence interval method along with its favorable properties in size and power. We finally illustrate those properties through carefully conducted simulations and real data examples.

\section{Methods}

We consider a noninferiority trial with treatment group ($T$) and control/standard-of-care group ($C$) having a binary endpoint representing whether or not an outcome is observed. Let $P_T$ and $P_C$ be the probabilities the outcome is observed, and let $\delta=P_T-P_C$ represent the risk difference. Depending on whether we are considering a positive outcome (e.g. resolution of a disease) or a negative outcome (e.g. cancer recurrence), we will use the following hypotheses for a noninferiority trial with pre-specified noninferiority margin $\delta_0\ge 0$.

\begin{table}[H]
    \caption{Description of the noninferiority hypotheses}
    \label{tab:hypothesis}    
    \centering
\begin{tabular}{c|c|c|l}
     Hypothesis & Positive Outcome & Negative Outcome & Interpretation  \\
     \hline 
     $H_0$ & $\delta\le -\delta_0$ &$\delta \ge \delta_0$ & ``inferior trial''; $T$ is inferior to $C$\\
     $H_1$ & $\delta> -\delta_0$ &$\delta < \delta_0$ & ``non-inferior trial''; $T$ is not inferior to $C$\\
     
\end{tabular}

\end{table}

We'll consider a positive outcome for the rest of this paper.   We model the binary outcomes of the treatment and control groups with the following binomial distributions.
\begin{align*}
&X_T\sim\text{Binomial}(N_T,P_T)\\ 
&X_C\sim\text{Binomial}(N_C,P_C) 
\end{align*}

Under this binomial model, the joint probability for $X_T=x_T$ and $X_C=x_C$ is \begin{align*}
\Pr(X_T=x_T,X_C=x_c\mid P_T, P_C)&= \binom{N_T}{x_T}\binom{N_C}{x_C} (P_T)^{x_T}(1-P_T)^{N_T-x_T} P_C^{x_C} (1-P_C)^{N_C-x_C},
\end{align*}
where $0\le x_T\le N_T$ and $0\le x_C\le N_C$. The likelihood can be reformulated in terms of $P_T$ and $\delta$ with the substitution $P_C=P_T-\delta$.
\begin{eqnarray}
\label{eq:likelihood}
\Pr(X_T=x_T,X_C=x_c\mid P_T, \delta)=\binom{N_T}{x_T}\binom{N_C}{x_C} (P_T)^{x_T}(1-P_T)^{N_T-x_T} (P_T-\delta)^{x_C} (1+\delta-P_T)^{N_C-x_C},
\end{eqnarray}
where $P_T$ and $\delta$ must satisfy the condition
\[
\max(0,\delta)\le P_T \le \min(1,1+\delta).
\]

If $\delta_0=0$, then $X_T+X_C$ is a sufficient statistic for $P_T$ under the null hypothesis, which forms the basis of Fisher's exact test procedure \citep{Fisher:35}. However, for the more general setup of $\delta_0\not=0$, a different approach is needed. One such approach is described in the next section.

\subsection{Chan's exact test}
\label{sec:chan98}

\cite{Chan:98}, and subsequently \cite{Rohmel:99}, proposed an unconditional exact p-value approach based on the maximization/minimax principle. This approach starts with specifying a preorder on the sample space
\[
\Omega=\left\{(x_T,x_C): 0\le x_T\le N_T \quad \text{and}\quad 0\le x_C\le N_C\right\}.
\]
More specifically, given the preorder, we can index and arrange the $n$ elements of $\Omega$ in the following manner
\[
\omega_1\le \omega_2\le \cdots \le \omega_n.
\]
A natural approach to specifying a preorder is to use a test statistic, or more generally any function, that maps the elements of $\Omega$ to $\mathbb{R}$. Let $S$ be a statistic that induces a preorder on $\Omega$. Without loss of generality, suppose that greater values of $S$ favor the alternative hypothesis (otherwise replace $S$ with $-S$).  This statistic will likely depend on $N_T$ and $N_C$ and may also depend on $\delta_0$. 
The so-called exact unconditional p-value, as described in \cite{Chan:98}, \cite{Rohmel:99}, and \cite{Chan:03}, can be expressed as
\begin{equation}
\label{eq:chan.pvalue}
p_S^{\text{exact}}(x_T,x_C)=\sup_{P_T,\delta\le -\delta_0}\Pr\left[S(X_T,X_C)\ge S(x_T,x_C) \mid P_T,\delta\right].
\end{equation}
This approach of uses the maximization/minimax principle \citep{Basu:11,Lehmann:06} to eliminate the nuisance parameters $P_T$ and $\delta$. Taking the supremum over both $P_T$ and $\delta$ can be simplified when the statistic $S$ satisfies the so-called Barnard criteria, stemming from \cite{Barnard:47}, which is given by the following two conditions:
\begin{equation}
\label{eq:Barnard}
\begin{aligned}
    S(x_T,x_C)&\ge S(x_T,x_C+1) \text{ for all }(x_T,x_C),(x_T,x_C+1)\in\Omega\\
    S(x_T,x_C)&\ge S(x_T-1,x_C) \text{ for all }(x_T,x_C),(x_T-1,x_C)\in\Omega
\end{aligned}
\end{equation}
This condition is intuitively clear: for any observed outcome, having one more success in the control or one less success in the treatment should lead to a smaller value of the test statistic. \cite{Rohmel:99} proved that when the inequalities in \eqref{eq:Barnard} are satisfied, the supremum in \eqref{eq:chan.pvalue} occurs on the boundary of $H_0$; i.e. the supremum is the maximum under the restriction $P_T-P_C=-\delta_0$. \cite{Frick:00} generalized this and proved that if either inequality in \eqref{eq:Barnard} is satisfied, then again the supremum in \eqref{eq:chan.pvalue} occurs at the boundary of $H_0$. We shall assume the statistic $S$ satisfies \eqref{eq:Barnard}, thus allowing us to rewrite Equation \eqref{eq:chan.pvalue} as follows.
\begin{equation}
\label{eq:chan.pvalue2}
p_S^{\text{exact}}(x_T,x_C)=\max_{P_T\in[0,1-\delta_0]}\Pr\left[S(X_T,X_C)\ge S(x_T,x_C) \mid P_T, \delta=-\delta_0\right]
\end{equation}
For a general statistic $p$ and specified level $\alpha$, we define the critical region, $\Omega_\alpha$, to be the set of elements of $\Omega$ that reject the null hypothesis based on a level $\alpha$ test; i.e.
\begin{equation}
\label{eq:CR}
\Omega_\alpha=\left\{(x_T,x_C)\text{ such that } p(x_T,x_C)\le \alpha\right\}. 
\end{equation}
Conditioned on $P_T$ and $\delta$, we define the conditional size of $p$ under the null to be
\begin{equation}
\label{eq:size2}
\alpha(p\mid P_T, \delta) = \;\smashoperator{\sum_{(x_T,x_C) \in \Omega_\alpha}}\; \Pr\left[X_T=x_T,X_C=x_C \mid P_T, \delta\right], 
\end{equation}
and the maximal size is defined as
\begin{equation}
\label{eq:size}
\alpha^*(p) = \sup_{P_T,\delta\le -\delta_0} \alpha(p\mid P_T,\delta).
\end{equation}


Following terminology of \cite{Berger:94} and \cite{Rohmel:99a}, we will call $p$ a \emph{valid p-value} if 
\begin{equation}
\label{eq:valid}
\alpha^*(p) \le \alpha \text{ for all }\alpha\in[0,1].
\end{equation}
In the following theorem, we show that $p^{\text{exact}}_S(x_T,x_C)$ is a valid p-value.

\begin{theorem}
  Let $p^{\text{exact}}_S$ be the exact unconditional p-value given in Equation \eqref{eq:chan.pvalue}. Then $p^{\text{exact}}_S$ is a valid p-value; i.e.
  \[
  \alpha^*\left(p^{\text{exact}}\right) \le \alpha \text{ for all }\alpha\in[0,1].
  \]
\end{theorem}
\begin{proof}
Let $\alpha\in[0,1]$, and let $\Omega_\alpha$ be the critical region for $p^{\text{exact}}$ as defined in Equation \eqref{eq:CR}. Let $(x_T^\alpha,x_C^\alpha)$ be a minimal element of $\Omega_\alpha$; i.e.
\[
S(x_T,x_C)\ge S(x_T^\alpha,x_C^\alpha) \text{ for all }(x_T,x_C)\in\Omega_\alpha.
\]
It is noted that this minimal element may not be unique. Define $\tilde\Omega_\alpha$ as follows
\[
\tilde\Omega_\alpha=\left\{(x_T,x_C)\in\Omega \text{ such that } S(x_T,x_C)\ge S(x_T^\alpha,x_C^\alpha)\right\}.
\]
If $(x_T,x_C)\in\Omega_\alpha$, then $S(x_T,x_C)\ge S(x_T^\alpha,x_C^\alpha)$, as $(x_T^\alpha,x_C^\alpha)$ is a minimal element, which implies $\Omega_\alpha\subset \tilde\Omega_\alpha$. Therefore,
\begin{align*}
\sup_{P_T,\delta\le -\delta_0}\left[\smashoperator[r]{\sum_{(x_T,x_C) \in \Omega_\alpha}}\; \Pr\left[X_T=x_T,X_C=x_C \mid P_T, \delta\right]\right]
&\le \sup_{P_T,\delta\le -\delta_0}\left[\smashoperator[r]{\sum_{(x_T,x_C) \in \tilde\Omega_\alpha}}\; \Pr\left[X_T=x_T,X_C=x_C \mid P_T, \delta\right]\right]\\
&= p_S^{\text{exact}}(x_T^\alpha,x_C^\alpha)\\
&\le \alpha
\end{align*}
\end{proof}
Ultimately, we will propose a confidence interval for $\delta$ that corresponds to $p^{\text{exact}}_S$ for a particular choice of $S$ -- the so-called $\delta$-projected Z-score --  which is described in the next section.

\subsection{$\delta$-projected Z-score}

There are several choices of statistics to define the preorder in Chan's method including 
\cite{Dunnett:77,Santner:80,Blackwelder:82,Miettinen:85,Farrington:90,Chan:99,Rohmel:99}. \cite{Chan:03,Chan:98} is particularly favorable to what he calls the $\delta$-projected Z-score, originally described in \cite{Miettinen:85}, given by
\begin{equation}
\label{eq:Zdelta}
Z_\delta(X_T,X_C)=\frac{\hat{P}_T-\hat{P}_C+\delta}{\hat\sigma_{\delta}}, 
\end{equation}
where
\[
\hat P_T=\frac{X_T}{N_T},\quad
\hat P_C=\frac{X_C}{N_C},\quad
\hat\sigma_{\delta}=\sqrt{\frac{\tilde P_T(1-\tilde P_T)}{N_T} + \frac{\tilde P_C(1-\tilde P_C)}{N_C}},
\]
and $\tilde P_T$ and $\tilde P_C$ represent the maximum likelihood estimators of $P_T$ and $P_C$, respectively, under the null hypothesis constraint $\tilde P_T-\tilde P_C=-\delta$. In particular, in calculating the exact p-value with Equation \eqref{eq:chan.pvalue2}, \cite{Chan:03,Chan:98} advocates the use of the statistic $S(X_T,X_C)=Z_{\delta_0}(X_T,X_C)$. We will simply write $p^{\text{exact}}(x_T,x_C)$ to refer to Chan's exact p-value with this statistic; i.e., 
\begin{equation}
\label{eq:pexact}
p^{\text{exact}}(x_T,x_C) = \max_{P_T\in[0,1-\delta_0]}\Pr\left[Z_{\delta_0}(X_T,X_C)\ge Z_{\delta_0}(x_T,x_C) \mid P_T, \delta=-\delta_0\right].
\end{equation}
\cite{Chan:99a,Chan:03} has provided justification that  $Z_{\delta_0}(X_T,X_C)$  satisfies the Barnard criteria (conditions in Equation \eqref{eq:Barnard}), so as in Equation \eqref{eq:chan.pvalue2}, the maximization in Equation \eqref{eq:pexact} occurs on the boundary of $H_0$. Closed formulas for calculating the restricted maximum likelihood estimators $\tilde P_T$ and $\tilde P_C$ are given in \cite{Miettinen:85} and  \cite{Farrington:90}.  

Asymptotically, $Z_{-\delta}(X_T,X_C)$ has a standard normal distribution, so it can be used as an asymptotic pivotal quantity to form a confidence set for $\delta$ \citep{Miettinen:85}. In particular, if $Z_\delta$ is monotonic in $\delta$, this confidence set would be a contiguous interval. We do not prove monotonicity due to the complexity of the statistic, but monotonicity is easy to confirm for any given set of parameters. We numerically confirmed $Z_\delta(x_T,x_C)$ is monotonically increasing for all tables up to $N_T=100$ and $N_C=100$ with a grid size on $\delta$ of 0.01. Therefore, in the discussion below, we simply assume $Z_\delta(x_T,x_C)$ is monotonically increasing in $\delta$.

The asymptotic $(1-\alpha)$ confidence interval of \cite{Miettinen:85} is $(\delta^\text{asy}_{L,\alpha},\delta^\text{asy}_{U,\alpha})$, where $\delta^\text{asy}_{L,\alpha}$ and $\delta^\text{asy}_{U,\alpha}$ satisfy
\begin{equation}
\label{eq:Z.confint}
  \mathfrak{z}_{1-\alpha/2}= Z_{-\delta^{\text{asy}}_{L,\alpha}}(x_T,x_C) \quad\text{and}\quad
  \mathfrak{z}_{\alpha/2} = Z_{-\delta^{\text{asy}}_{U,\alpha}}(x_T,x_C);
\end{equation}
the notation $\mathfrak{z}_{\alpha}$ represents the $\alpha$ quantile of the standard normal distribution. The following probability coverage calculation validates this asymptotic confidence interval formulation
\begin{align}
\Pr\left[\delta^\text{asy}_{L,\alpha} \le \delta \le \delta^\text{asy}_{U,\alpha}\right] &= \Pr\left[-\delta^\text{asy}_{U,\alpha} \le -\delta \le -\delta^\text{asy}_{L,\alpha}\right]\\
&=\Pr\left[Z_{-\delta^\text{asy}_{U,\alpha}}(X_T,X_C) \le Z_{-\delta}(X_T,X_C) \le Z_{-\delta^\text{asy}_{L,\alpha}}(X_T,X_C)\right]\\
&=\Pr\left[\mathfrak{z}_{\alpha/2}\le Z_{-\delta}(X_T,X_C) \le \mathfrak{z}_{1-\alpha/2}\right]\\
&\approx 1-\alpha
\end{align}

We define the asymptotic p-value to be
\begin{equation}
\label{eq:pasy}
p^{\text{asy}}(x_T,x_C)=\Pr\left[Z\ge Z_{\delta_0}(x_T,x_C)\right] = 1-\Phi\left( Z_{\delta_0}(x_T,x_C)\right),
\end{equation}
where $Z$ represents the standard normal distribution and $\Phi(x)$ is the cumulative distribution function for the standard normal distribution. Theorem \ref{thm:pasy} establishes a connection between $p^{\text{asy}}(x_T,x_C)$ and $\delta^\text{asy}_{L,\alpha}$.

\begin{theorem}
\label{thm:pasy}
Let $\delta_{L,\alpha}^\text{asy}$ be the asymptotic lower bound given in Equation \eqref{eq:Z.confint}, and let $p^\text{asy}(x_T,x_C)$ be the p-value defined in Equation \eqref{eq:pasy}. Assuming $Z_\delta(x_T,x_C)$ is monotonically increasing in $\delta$, then
\[
p^\text{asy}(x_T,x_C)\le \alpha/2 \quad \text{if and only if}\quad \delta_{L,\alpha}^\text{asy}(x_T,x_C)> -\delta_0.
\]
\end{theorem}
\begin{proof}
Suppose $p^\text{asy}(x_T,x_C)\le \alpha/2$. From Equation \eqref{eq:pasy}, we have
\[
1-\Phi\left(Z_{\delta_0}(x_T,x_C)\right)\le \alpha/2,
\]
which implies
\[
Z_{\delta_0}(x_T,x_C) \ge \Phi^{-1}\left(1-\alpha/2\right) = \mathfrak{z}_{1-\alpha/2}=Z_{-\delta^{\text{asy}}_{L,\alpha}}(x_T,x_C),
\]
using Equation \eqref{eq:Z.confint}. From the assumption that $Z_{\delta}(x_T,x_C)$ is monotonically increasing in $\delta$, we have that $\delta^{\text{asy}}_{L,\alpha} \ge -\delta_0$. 
The same steps can be used backwards to show that $\delta_{L,\alpha}^\text{asy}(x_T,x_C)> -\delta_0$ implies that $p^\text{asy}(x_T,x_C)\le \alpha/2$. 
\end{proof}

  We note that the confidence interval $(\delta^{\text{asy}}_{L,\alpha},\delta^{\text{asy}}_{U,\alpha})$, which is based on $Z_\delta$, does not depend on the prespecified value of $\delta_0$. It is also noted that $\alpha^*(p^\text{asy})$ is not necessarily bounded by $\alpha/2$; furthermore, there are many examples in which the type-I error exceeds the level $\alpha/2$ (see Section \ref{sec:examples} below). Hence, $p^\text{asy}$ is not a valid p-value per the definition stated above in Section \ref{sec:chan98}. Next we describe a confidence interval that utilizes Chan's exact p-value with guaranteed probability coverage.

\subsection{Chan \& Zhang confidence interval}
\cite{Chan:99} proposed an ``exact'' two-sided $(1-\alpha)\%$ confidence interval for the risk difference $\delta$. The method is based on inverting two one-sided hypotheses using the $\delta$-projected Z-score $Z_\delta$. We first define the following quantities
\begin{align*}
    P_{L,\delta}(x_T,x_C)&=\max_{P_T\in [0,1-\delta]} \Pr\left[Z_{-\delta}(X_T,X_C)\ge Z_{-\delta}(x_T,x_C)\right]\\ 
    P_{U,\delta}(x_T,x_C)&=\max_{P_T\in [0,1-\delta]} \Pr\left[Z_{-\delta}(X_T,X_C)\le Z_{-\delta}(x_T,x_C)\right].
\end{align*}
In particular, we note that the exact p-value, given in Equation \eqref{eq:pexact}, can be written as 
\[
p^{\text{exact}}(x_T,x_C) = P_{L,-\delta_0}(x_T,x_C).
\]

The Chan \& Zhang confidence interval, denoted $(\delta^{\text{CZ}}_{L,\alpha},\delta^{\text{CZ}}_{U,\alpha})$, is defined by the following expressions.
\begin{align}
    \delta^\text{CZ}_{L,\alpha}(x_T,x_C)&= \inf_{\delta} \left\{\delta:P_{L,\delta}(x_T,x_C) > \alpha/2\right\} \label{eq:CZ.lower} \\
    \delta^\text{CZ}_{U,\alpha}(x_T,x_C)&= \sup_{\delta} \left\{\delta:P_{U,\delta}(x_T,x_C) > \alpha/2\right\} \nonumber
\end{align}
It is noted that this confidence interval, like the asymptotic confidence interval in the previous section, does not depend on the noninferiority margin $\delta_0$. 
We correspond the Chan \& Zhang confidence interval with a Chan \& Zhang p-value defined as
\begin{equation}
\label{eq:p.CZ}
p^\text{CZ}(x_T,x_C)=\max_{\delta\in[-1,-\delta_0]} P_{L,\delta} (x_T,x_C),
\end{equation}
where the correspondence is established in Theorem \ref{thm:chan-zhang}.

\begin{theorem}
\label{thm:chan-zhang}
Let $\delta_{L,\alpha}^\text{CZ}(x_T,x_C)$ be the Chan \& Zhang lower bound given in Equation \eqref{eq:CZ.lower}, and let $p^\text{CZ}(x_T,x_C)$ be the Chan \& Zhang p-value given in Equation \eqref{eq:p.CZ}.  
  \begin{description}
  \item[(i)] $p^\text{CZ}(x_T,x_C)\le \alpha/2$ if and only if $\delta_{L,\alpha}^\text{CZ}(x_T,x_C)> -\delta_0$.
  \item[(ii)]  $p^\text{CZ}(x_T,x_C)$ is bounded below by $p^\text{exact}(x_T,x_C)$. In particular, the Chan \& Zhang p-value is valid and satisfies the following inequalities
\begin{align}
  &p^\text{CZ}(x_T,x_C)\ge p^\text{exact}(x_T,x_C) \label{eq:inequality1}\\
  &\alpha^*(p^\text{CZ}) \le \alpha^*(p^\text{exact}) \le \alpha/2.\label{eq:inequality2}
\end{align}
  \end{description}
\end{theorem}
\begin{proof}
Suppose 
\[
p^\text{CZ}(x_T,x_C) \triangleq \max_{\delta\in[-1,-\delta_0]} P_{L,\delta}(x_T,x_C) \le \alpha/2.
\]
Then
\[
\delta_{L,\alpha}^{\text{CZ}}(x_T,x_C) = \inf_{\delta} \left\{\delta:P_{L,\delta}(x_T,x_C) > \alpha/2\right\} > -\delta_0,
\]
which establishes $\delta_{L,\alpha}^{\text{CZ}}>-\delta_0$.

Now suppose 
\[
\delta_L^{\text{CZ}}(x_T,x_C) \triangleq \inf_{\delta} \left\{\delta:P_{L,\delta}(x_T,x_C) > \alpha/2\right\} > -\delta_0, 
\]
This implies
\[ 
\max_{\delta\in[-1,-\delta_0]} P_{L,\delta}(x_T,x_C) = p^\text{CZ}(x_T,x_C) < \alpha/2.
\]
This establishes part (i). Part (ii) immediately follows from 
\[
p^\text{CZ}(x_T,x_C) = \max_{\delta\in[-1,-\delta_0]} P_{L,\delta}(x_T,x_C) \ge  P_{L,-\delta_0}(x_T,x_C) = p^\text{exact}(x_T,x_C)
\]
and the definition of $\alpha^*(\cdot)$ provided in Equation \eqref{eq:size}. 
\end{proof}

Equation \eqref{eq:inequality1} in Theorem \ref{thm:chan-zhang} gives $p^\text{CZ}(x_T,x_C) \ge p^\text{exact}(x_T,x_C)$.  This is equivalent to the following statement:
\[
p^\text{CZ}(x_T,x_C) \text{ rejects $H_0$ }\Rightarrow p^\text{exact}(x_T,x_C) \text{ rejects $H_0$.}
\]
So anytime the Chan \& Zhang confidence interval rejects $H_0$, the exact test will also necessarily reject $H_0$, but the converse is not always true. This indicates the exact test will have at least as much statistical power as the test induced by the Chan \& Zhang confidence interval.

Figure \ref{fig:myplot_example} illustrates the relationship between $p^{\text{CZ}}(x_T,x_C)$ and $p^{\text{exact}}(x_T,x_C)$ with a concrete example taking $N_T=9$, $N_C=19$, $x_T=5$, and $x_C=10$. In the figure, the black line is $P_{L,\delta}(x_T,x_C)$ and the red line is $\max_{\delta\in[-1,-\delta]} P_{L,\delta} (x_T,x_C)$. The values for which these lines intersect at $\delta=-\delta_0$ correspond to $p^{\text{exact}}(x_T,x_C)$ and $p^{\text{CZ}}(x_T,x_C)$, respectively, as indicated on the figure with $\delta_0=0.1$. The regions shaded in blue indicate values of $\alpha/2$ and $\delta_0$ that would cause $p^{\text{exact}}(x_T,x_C)$ to reject $H_0$ and $p^{\text{CZ}}(x_T,x_C)$ not to reject $H_0$ for a level $\alpha$ test. Also illustrated on the graphic is the lower bound of the Chan \& Zhang confidence interval for $\alpha=0.5$. The region shaded in orange is not of interest as the corresponding $\alpha$ would be greater than 1 in this region.
\begin{figure}
    \centering
    \includegraphics[width=.75\textwidth]{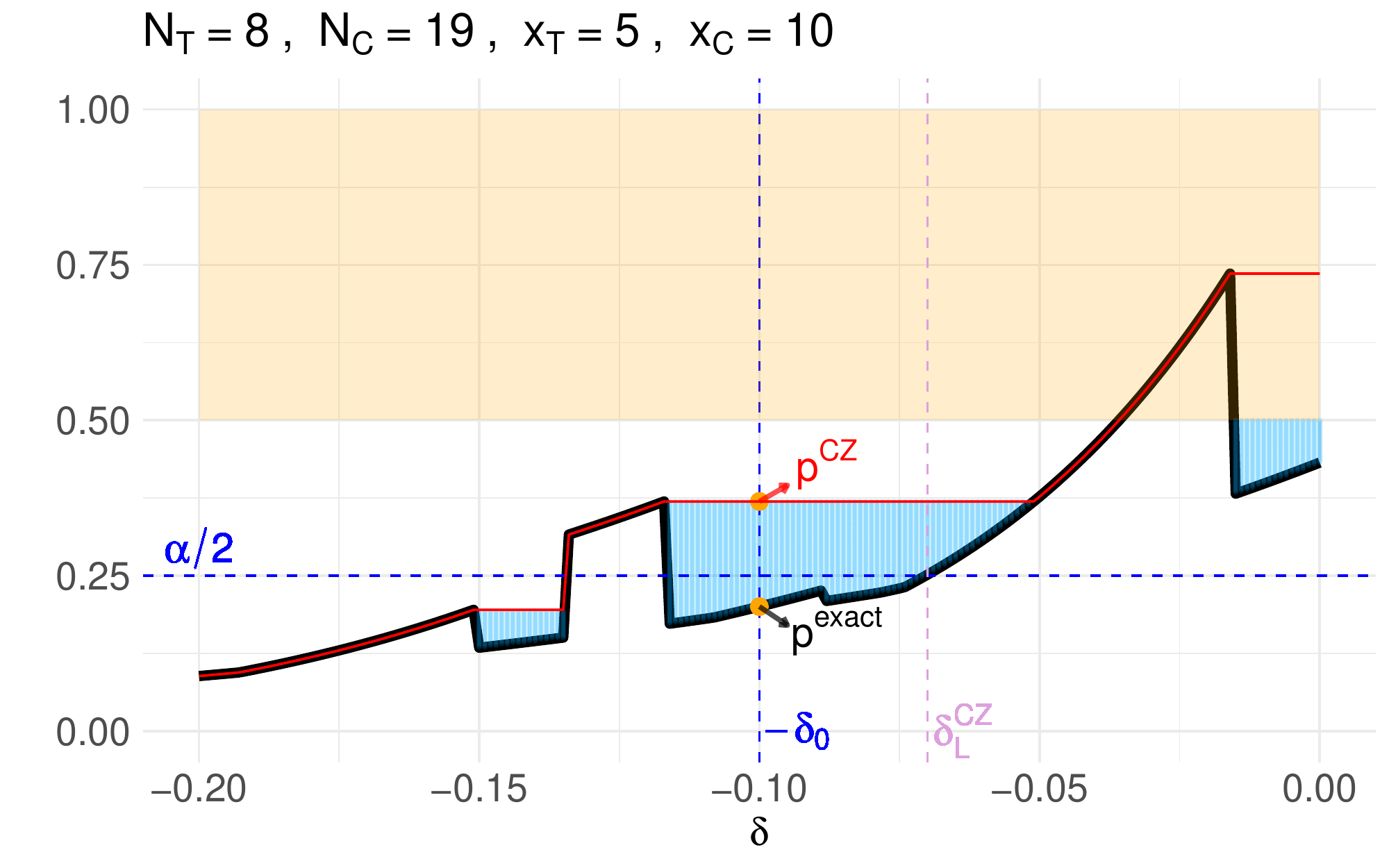}
    \caption{This graphic shows the relationship between $p^{\text{CZ}}(x_T,x_C)$ (red line) and $p^{\text{exact}}(x_T,x_C)$ (black line).  The regions shaded in blue indicate the values of $\alpha/2$ and $\delta_0$ what would cause $p^{\text{exact}}(x_T,x_C)$ to reject $H_0$ and $p^{\text{CZ}}(x_T,x_C)$ not to reject $H_0$ for a level $\alpha$ test. 
    }
    \label{fig:myplot_example}
\end{figure}

As illustrated in Figure \ref{fig:myplot_example}, there are many situations in which the strict inequality $p^\text{CZ}(x_T,x_C) < p^\text{exact}(x_T,x_C)$ holds. Using terminology described in  \cite{Rohmel:99}, we say that $p^\text{exact}$ \emph{strictly dominates} $p^\text{CZ}$. A p-value that is not strictly dominated is called \emph{acceptable}. It is much easier establishing a p-value is not acceptable, like $p^\text{CZ}$, than to prove a given p-value, say $p_S^\text{exact}$, is acceptable. \cite{Frick:00} provides various necessary and sufficient conditions for acceptable p-values. 

Next, we propose a novel ``exact-corrected'' confidence interval, $(\delta^{\text{EC}}_{L,\alpha},\delta^{\text{EC}}_{U,\alpha})$, that corresponds to $p^\text{exact}$; i.e. $p^\text{exact}(x_T,x_C)\le \alpha/2$ if and only if $\delta^{\text{EC}}_{L,\alpha}> -\delta_0$.

\subsection{Exact-corrected $\delta$-projected Z-score}

We consider a modification of the $\delta$-projected $Z$-score, which we call the exact-corrected (EC) $\delta$-projected $Z$-score. This exact-corrected $\delta$-projected $Z$-score, labeled $Z_\delta^{\text{EC}}(X_T,X_C)$, is given by the following expression. 
\begin{equation}
\label{eq:Ztilde}
\begin{aligned}
Z^{\text{EC}}_\delta(X_T,X_C)&=
\frac{\hat{P}_T-\hat{P}_C+\delta}{\hat\sigma_{\delta}}-\frac{\hat\sigma_{\delta_0}}{\hat\sigma_{\delta}} \left(Z_{\delta_0}(X_T,X_C)-\Phi^{-1}\left(1-p^{\text{exact}}(X_T,X_C)\right)\right)\\
&=
\frac{\hat{P}_T-\hat{P}_C+\delta}{\hat\sigma_{\delta}}-\frac{\hat\sigma_{\delta_0}}{\hat\sigma_{\delta}} \left(\Phi^{-1}\left(1-p^{\text{asy}}(X_T,X_C)\right)-\Phi^{-1}\left(1-p^{\text{exact}}(X_T,X_C)\right)\right)\\
&= Z_\delta(X_T,X_C) - \mathrm{EC}_{\delta}(X_T,X_C)
\end{aligned}
\end{equation}
where $\Phi^{-1}$ denotes the quantile function of the standard normal distribution (also called the probit function), and $\mathrm{EC}_{\delta}(X_T,X_C)$ denotes the exact correction term given by
\[
\mathrm{EC}_{\delta}(X_T,X_C)=\frac{\hat\sigma_{\delta_0}}{\hat\sigma_{\delta}} \left(\Phi^{-1}\left(1-p^{\text{asy}}(X_T,X_C)\right)-\Phi^{-1}\left(1-p^{\text{exact}}(X_T,X_C)\right)\right).
\]
In particular, when evaluating $Z^{\text{EC}}_\delta(X_T,X_C)$ at $\delta=\delta_0$, $X_T=x_T$, and $X_C=x_C$, we have
\[
Z^{\text{EC}}_{\delta_0}\left(x_T,x_C\right)=\Phi^{-1}\left(1-p^{\text{exact}}(x_T,x_C)\right).
\]

 In the simulation section, we numerically show the expectation of  $\mathrm{EC}_{\delta}(X_T,X_C)$ tends to zero with increasing sample size over selected values of $P_T$, $P_C$, and $\delta_0$. In the subsequent discussion, we assume $Z_\delta^{\text{EC}}(x_T,x_C)$ is monotonic in $\delta$, so inverting $Z_\delta^{\text{EC}}(x_T,x_C)$ will produce $(1-\alpha)\%$ ``exact-corrected'' confidence interval, $(\delta^{\text{EC}}_{L,\alpha},\delta^\text{EC}_{U,\alpha})$, defined by the following equations

\begin{equation}
\label{eq:confint.EC}
  \mathfrak{z}_{1-\alpha/2}= Z_{-\delta^{\text{EC}}_{L,\alpha}}(X_T,X_C)\quad\text{and}\quad
\mathfrak{z}_{\alpha/2}= Z_{-\delta^{\text{EC}}_{U,\alpha}}(X_T,X_C).
\end{equation}

\begin{theorem}
\label{thm:EC}
Let $\delta_{L,\alpha}^\text{EC}$ be the exact-corrected lower bound given in Equation \eqref{eq:confint.EC}, and let $p^\text{exact}(x_T,x_C)$ be Chan's exact p-value based on the $\delta$-projected $Z$-score as given in Equation \eqref{eq:pexact}. Assuming $Z_\delta^\text{EC}(x_T,x_C)$ is monotonically increasing in $\delta$, then
\[
p^\text{exact}(x_T,x_C)\le \alpha/2 \quad \text{if and only if}\quad \delta_{L,\alpha}^\text{EC}(x_T,x_C)> -\delta_0.
\]
\end{theorem}
\begin{proof}
Note that 
\begin{equation}
\label{eq:pchain}
\Pr\left[Z\ge  Z^{\text{EC}}_{\delta_0}(x_T,x_C)\right] = \Pr\left[Z\ge \Phi^{-1}\left(1-p^{\text{exact}}\right)\right] =1- \Phi \left(\Phi^{-1}\left(1-p^{\text{exact}}(x_T,x_C)\right)\right) = p^\text{exact}(x_T,x_C).
\end{equation}

Suppose $p^\text{exact}(x_T,x_C)\le \alpha/2$. From Equation \eqref{eq:pchain}, we have
\[
\Pr\left[Z\ge  Z^{\text{EC}}_{\delta_0}(x_T,x_C)\right]=1-\Phi\left(Z_{\delta_0}^{\text{EC}}(x_T,x_C)\right)\le \alpha/2.
\]
This implies
\[
Z_{\delta_0}^{\text{EC}}(x_T,x_C) \ge \Phi^{-1}\left(1-\alpha/2\right) = z_{1-\alpha/2}=Z_{-\delta^{\text{EC}}_{L,\alpha}}(x_T,x_C),
\]
using Equation \eqref{eq:confint.EC}. From the assumption $Z^{\text{EC}}_{\delta_0}(x_T,x_C)$ is monotonically increasing in $\delta$, we have that $\delta^{\text{EC}}_{L,\alpha} \ge -\delta_0$. 
The same steps can be used backwards to show that $\delta_{L,\alpha}^\text{EC}(x_T,x_C)> -\delta_0$ implies that $p^\text{exact}(x_T,x_C)\le \alpha/2$. 
\end{proof}

\section{Simulations and Examples}

\subsection{Software and data sharing}

The R package \texttt{EC}, available on github at \url{https://github.com/NourHawila/EC}, allows the user to easily compute the confidence intervals, p-values, and maximal sizes discussed in this paper. All data that support the findings of this research are provided within the paper. 

\subsection{Asymptotic assessments}

Here we present simulations to show the asymptotic behavior of the expected value of $\mathrm{EC}_{\delta}(X_T,X_C)$. Seven different sample sizes of $N=N_T=N_C$, doubling each time from 10 to 640, were considered along with three different values of $\delta_0$ (0, 0.1, 0.2), three different values of $P_T$ (0.3, 0.5, 0.7), and three different values of $P_C$ (0.3, 0.5, 0.7). Each expected value is computed over 10,000 realizations of the data, thus yielding very precise estimates. Figure \ref{fig:myplot_asymptotics} shows $\E\left[\mathrm{EC}_{\delta}(X_T,X_C)\right]\approx 0$ for large values of $N$. This, in turn, suggests that $Z^{\text{EC}}_\delta(X_T,X_C)$ is close to $Z_\delta(X_T,X_C)$ for large $N$.

\begin{figure}
    \centering
    \includegraphics{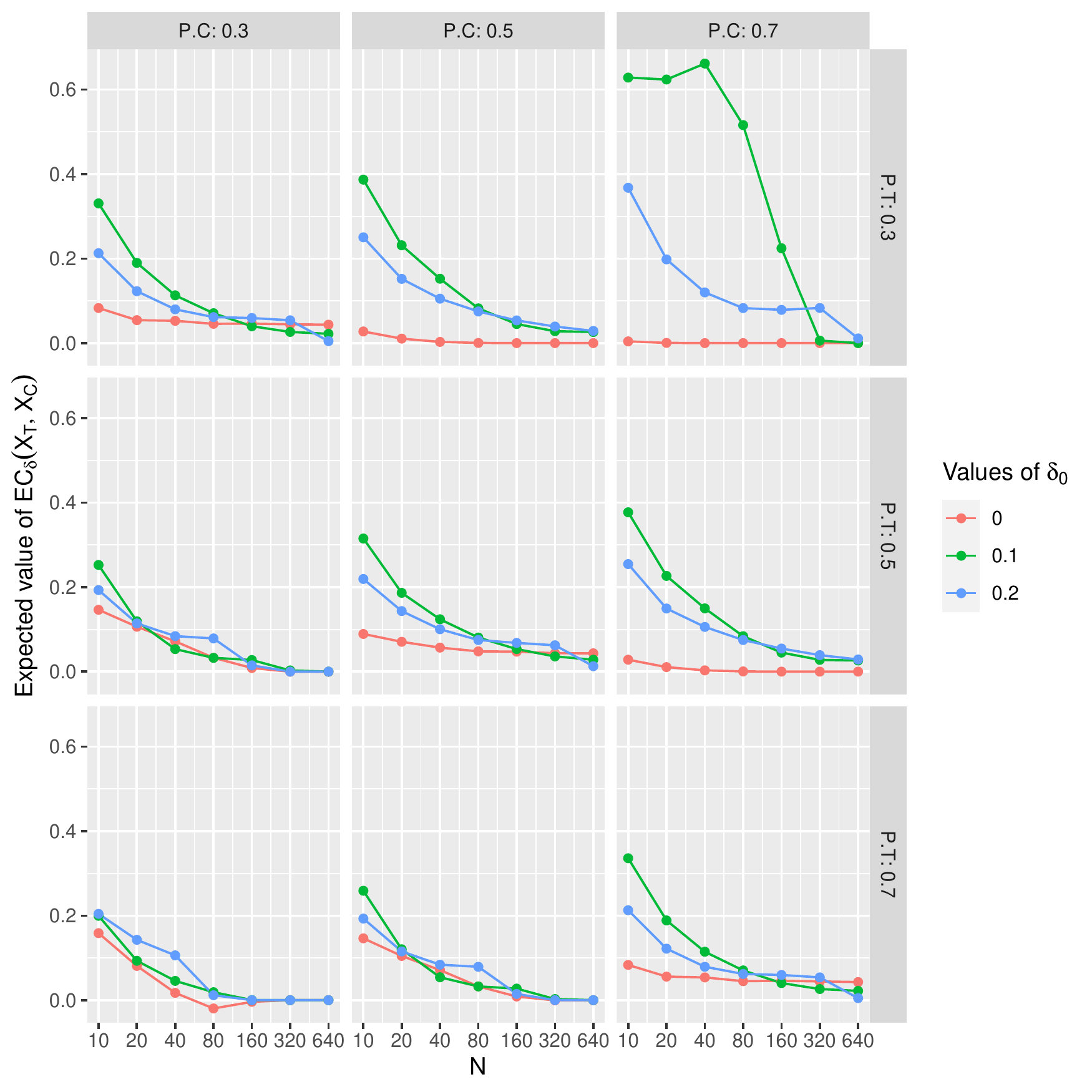}
    \caption{Expected value of $\mathrm{EC}_{\delta}(X_T,X_C)$ is approximated for different combinations of $P_T$, $P_C$, $\delta_0$, and $N=N_T=N_C$.}
    \label{fig:myplot_asymptotics}
\end{figure}

\subsection{Power and size}
\label{sec:power}
The performance of the confidence interval estimators are compared in terms of power and size. 
\begin{itemize}
\item The method ``MN'' \citep{Miettinen:85} corresponds to $(\delta^{\text{asy}}_{L,\alpha},\delta^{\text{asy}}_{U,\alpha})$ 
\item The method ``CZ'' \citep{Chan:99} corresponds to $(\delta^{\text{CZ}}_{L,\alpha},\delta^{\text{CZ}}_{U,\alpha})$.
\item The method ``EC'' corresponds to our proposed ``exact-corrected'' confidence interval estimator $(\delta^{\text{EC}}_{L,\alpha},\delta^{\text{EC}}_{U,\alpha})$.
\end{itemize}

The two examples displayed in Figure \ref{fig:myplot_power} showcase the potential differences in power and size across the three methods.  Once the values of $P_T$, $N_T$, $N_C$, $\delta_0$, and $\alpha$ are determined, the probability of rejecting $H_0$ for different values of $\delta$ is calculated from the likelihoods of the $(N_T+1)(N_C+1)$ tables using Equation \eqref{eq:likelihood}.  Values of $\delta$ that are smaller than $-\delta_0$ correspond to $H_0$ being true, whereas values of $\delta$ that are larger than $-\delta_0$ correspond to $H_1$ being true. The values \texttt{n.AA}, \texttt{n.AR}, and \texttt{n.RR} displayed on the graphic represent the number of the $(N_T+1)(N_C+1)$ tables for which CZ and EC both accept $H_0$ (\texttt{n.AA}), CZ accepts $H_0$ and EC rejects $H_0$ (\texttt{n.AR}), and CZ and EC both reject $H_0$ (\texttt{n.RR}). Note that the term ``accept'' here is used synonymously with ``failed to reject''.

Figure \ref{fig:power1} sets $P_T=0.95$, $N_T=5$, $N_C=11$, $\delta_0=0.03$, and $\alpha=0.7$. In this example there are four tables for which EC rejects $H_0$ but CZ fails to reject $H_0$. This causes EC to have greater power compared to CZ, though both methods have controlled size under $H_0$. We also see the MN method has better power than both CZ and EC, but also rejects $H_0$ with probability greater than $\alpha$ when $H_0$ is true. 

Figure \ref{fig:power2} sets $P_T=0.1$, $N_T=12$, $N_C=5$, $\delta_0=0.33$, and $\alpha=0.1$. In this example, there is only one table for which EC rejects $H_0$ but CZ fails to reject $H_0$, yet this one table occurs with high enough probability to produce a measurable difference in power between the EC and CZ methods. The MN and EC methods reject/accept $H_0$ for all tables (even though they produce different confidence intervals) causing them to have identical power curves. In this example, all methods have type-I error that is bounded by $\alpha$.

\begin{figure}
    \centering
    \begin{subfigure}{.49\textwidth}
    \includegraphics[width=\textwidth]{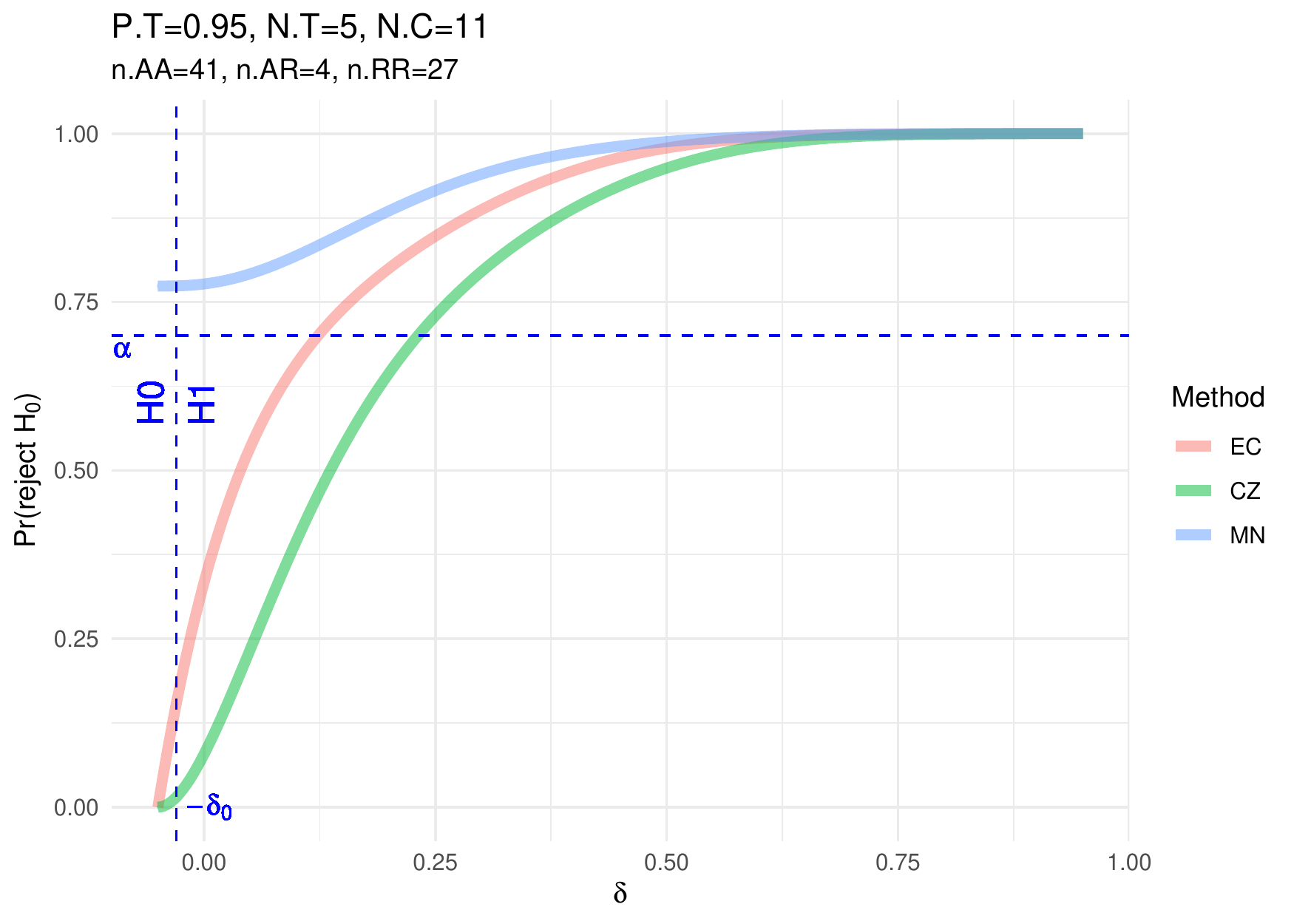}
     \caption{}
    \label{fig:power1}
    \end{subfigure}
    \hfill
    \begin{subfigure}{.49\textwidth}
    \includegraphics[width=\textwidth]{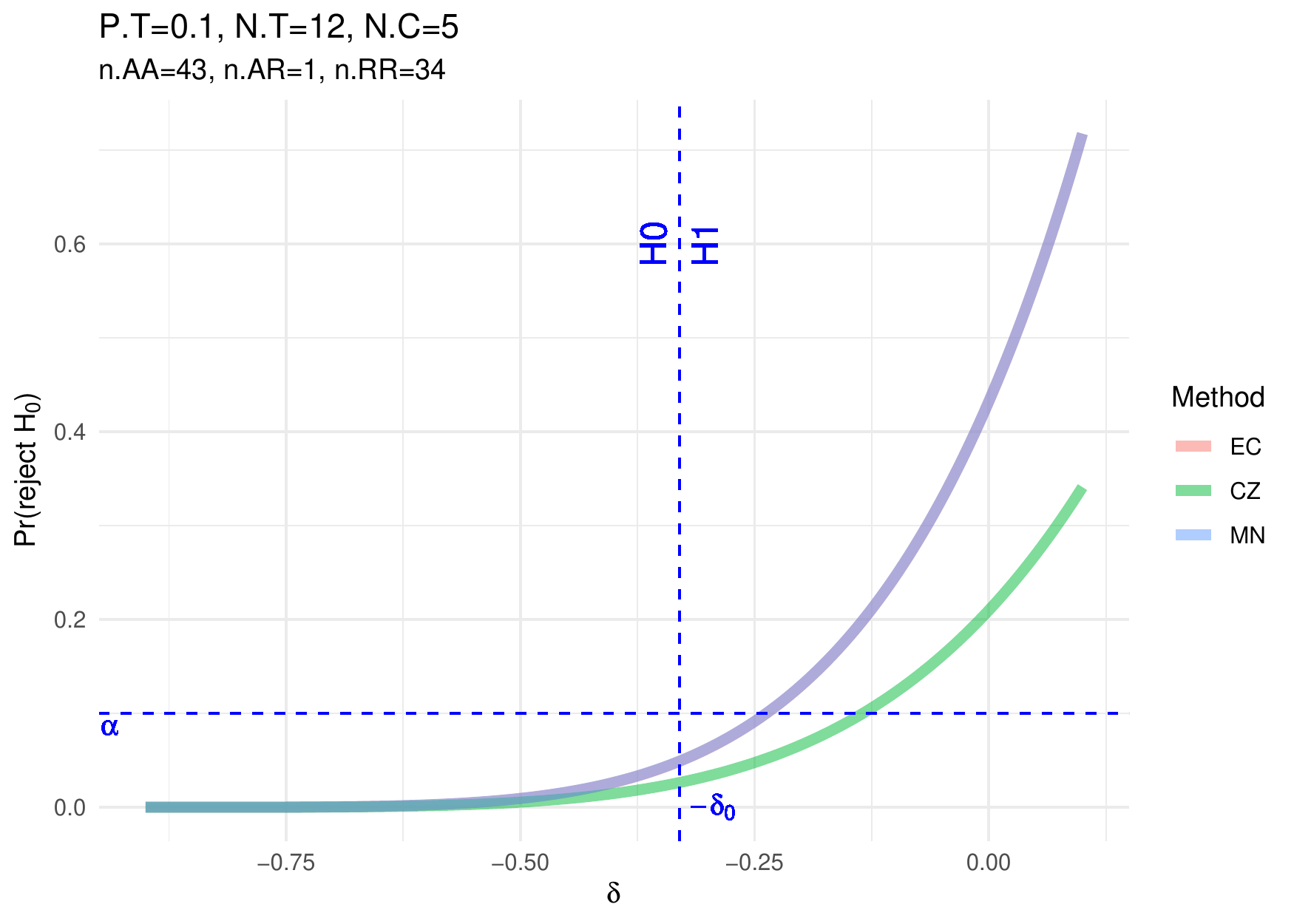}
    \caption{}
    \label{fig:power2}
    \end{subfigure}
    \caption{The probability of rejecting $H_0$ over different values of $\delta$ for the three confidence interval methods EC, CZ, and MN.}
    \label{fig:myplot_power}
\end{figure}
\newpage

\subsection{Data examples}
\label{sec:examples}

In addition to the three previously discussed methods -- EC, CZ, and MN -- we also consider the commonly used Wald's method \citep{Altman:13,Fagerland:15}. The Wald Z-statistic is given by 
\[
Z^\text{Wald}=\frac{\hat{P}_T-\hat{P}_C+\delta_0}{\sqrt{\frac{\hat {P}_T(1-\hat {P}_T)}{N_T} + \frac{\hat {P}_C(1-\hat {P}_C)}{N_C}}},
\]
and the corresponding confidence interval and p-value are given by
\begin{align*}
&\hat{P}_T-\hat{P}_C\pm z_{1-\alpha/2} \sqrt{\frac{\hat {P}_T(1-\hat {P}_T)}{N_T} + \frac{\hat {P}_C(1-\hat {P}_C)}{N_C}}\\
&p^\text{Wald} = 1- \Phi(Z^\text{Wald}).
\end{align*}

We first present three examples in which the EC and CZ confidence intervals produce different hypothesis test decisions. As shown in Theorem \ref{thm:chan-zhang}, $p^\text{exact}\le p^{CZ}$, so if the hypothesis test decisions differ between EC and CZ, it must be that EC rejects the null and CZ fails to reject the null. The parameters for the first example are the same as the parameters presented in Figure \ref{fig:myplot_example}. The second and third examples also show advantages of the EC method over the CZ method but with the more standard $\alpha=0.05$. Confidence intervals for all four methods are presented in Figure \ref{fig:ci}. Table \ref{tab:examples} shows the associated p-values and maximal sizes for the four methods. Consistent with Theorems \ref{thm:pasy}, \ref{thm:chan-zhang}, and \ref{thm:EC}, we see that the associated p-values are shown to declare noninferiority (reject the null with the corresponding p-value being less than $\alpha/2$) if and only if the lower bound of the respective confidence intervals is bigger than $-\delta_0$. 

Theorem \ref{thm:chan-zhang} also establishes the following inequalities on the maximal sizes of the CZ and EC methods:
\[\alpha^*(p^\text{CZ}) \le \alpha^*(p^\text{exact}) \le \alpha/2.\]
Note that maximal size does not depend on specific values of $x_T$ and $x_C$. The maximal size calculations shown in Table \ref{tab:examples} demonstrate this conservativeness of the CZ method over the EC method. Additionally, in each of these three examples, the maximal sizes for both the MN and Wald methods exceed the $\alpha/2$ threshold, which shows that the type-I error rates associated with confidence intervals produced from the MN and Wald methods can be inflated.


\begin{table}[H]
     \caption{Three examples show the EC method is less conservative than the CZ method yet still corresponding to an exact test that controls type-I error. }
     \label{tab:examples}    
     \centering
 \resizebox{\textwidth}{!}{%
 \begin{tabular}{r c c c c      c     c c c c    c     c c c c }
 \hline
  \multicolumn{1}{c}{} & \multicolumn{4}{c}{data parameters} & & \multicolumn{4}{c}{p-values} & & \multicolumn{4}{c}{maximal size}\\
  \cline{2-5} \cline{7-10} \cline{12-15}
       & $x_T/N_T$ & $x_C/N_C$ & $\delta_0$ & $\alpha/2$  & & EC & CZ & MN & Wald  & & EC & CZ & MN & Wald  \\
      \hline 
      Example 1 & 5/8 & 10/19 & 0.10 & 0.25  & & 0.200 & 0.370 & 0.172 & 0.167 & & 0.197 & 0.197 & 0.430 & 0.430\\
      Example 2 & 5/6 & 2/6 & 0.12  & 0.025 & & 0.023 & 0.030 & 0.014 & 0.006  & & 0.022 & 0.012 & 0.030 & 0.464\\
      Example 3 & 7/18 & 5/25 & 0.10 & 0.025 & & 0.024 & 0.027 & 0.018 & 0.020 & & 0.024 & 0.021 & 0.028 & 0.150\\
      \hline
 \end{tabular}
 }
\end{table}

In the next three examples, we compare the different methods from published studies. The first example, originally published in \cite{Lemerle:83} and later reanalyzed in \cite{Rodary:89} and \cite{Chan:98}, considers a randomized trial in childhood nephroblastoma comparing  a neoadjuvant chemotherapy (treatment) to radiation therapy (control) with the outcome of preventing tumor rupture during surgery.  A noninferiority margin is taken to be $\delta_0=0.1$, and the chemotherapy treatment would be considered non-inferior to radiation if $\delta=P_T-P_C>-0.1$. 83 of the 88 chemotherapy subjects had a positive outcome ($\hat P_T=.943$), and 69 of the 76 radiation subjects had a positive outcome ($\hat P_C=.908$). There is pretty strong evidence that neoadjuvant chemotherapy is not inferior to radiation in preventing surgical tumor rupture in this study.

The second example, originally published in \cite{Fries:93} and later reanalyzed in \cite{Chan:98}, considers the protective efficacy against illness of a recombinant protein flu vaccine in response to exposure to the H1N1 virus. The noninferiority margin is set to $\delta_0=0$, so the vaccine would be considered meaningful if $\delta=P_T-P_C>0$. 8 out of 15 subjects who received the vaccine (treatment group) avoided any kind of clinical illness ($\hat P_T=0.533$), whereas only 3 out of 15 subjects who received the placebo (control group) avoided illness ($\hat P_C=0.200$). Even with the small sample size, this study gives pretty strong evidence that the recombinant protein vaccine is more effective than placebo in preventing illness from exposure to H1N1. 

The third example, published in \cite{Kim:13}, considers whether the success rate of subclavian venous catheterization using a neutral shoulder position (treatment group) is not inferior to the often recommended retracted shoulder position (control group). The noninferiority margin is set to $\delta_0=0.05$, so the neutral shoulder position would be non-inferior if $\delta=P_T-P_C>-0.05$. 173 out of 181 subjects in the neutral position had a successful cathereterization ($\hat P_T=0.956$), and 174 out of 181 subjects in the retracted position had a successful cathereterization ($\hat P_C=0.961$). The success rates in this study are quite comparable for the two groups. We note that this study reports a confidence interval for $\delta$ based on Wald's method, which we show produces a decision that is inconsistent with the other three methods. 

The four confidence interval methods for each of these three examples with $\alpha=0.05$ are presented in Figure \ref{fig:ci}. In the Rodary et al.~example, all confidence intervals have a lower bound around $-0.05$ declaring noninferiority for the chemotherapy treatment. In the Fries et al.~example, EC, CZ, and MN intervals report a different decision compared to the Wald interval. The Fisher's exact and Chan's exact p-values for this example are 0.128 and 0.008, respectively, thus leading to a different statistical decision based on a 5\% level test. The lower bound of the EC and CZ intervals match, but the upper bound of the EC interval is somewhat smaller. The Kim et al.~example has relatively large sample sizes, so the EC, CZ, and MN methods produce confidence intervals that are all fairly similar to each other. However, these three methods produce a different conclusion about noninferiority compared to the Wald interval that was reported in the paper. That is, the EC, CZ, and MN intervals fail to conclude noninferiotiy at the 5\% margin, whereas the Wald interval establishes noninferiority. 


\begin{figure}
    \centering
    \includegraphics[width=\textwidth]{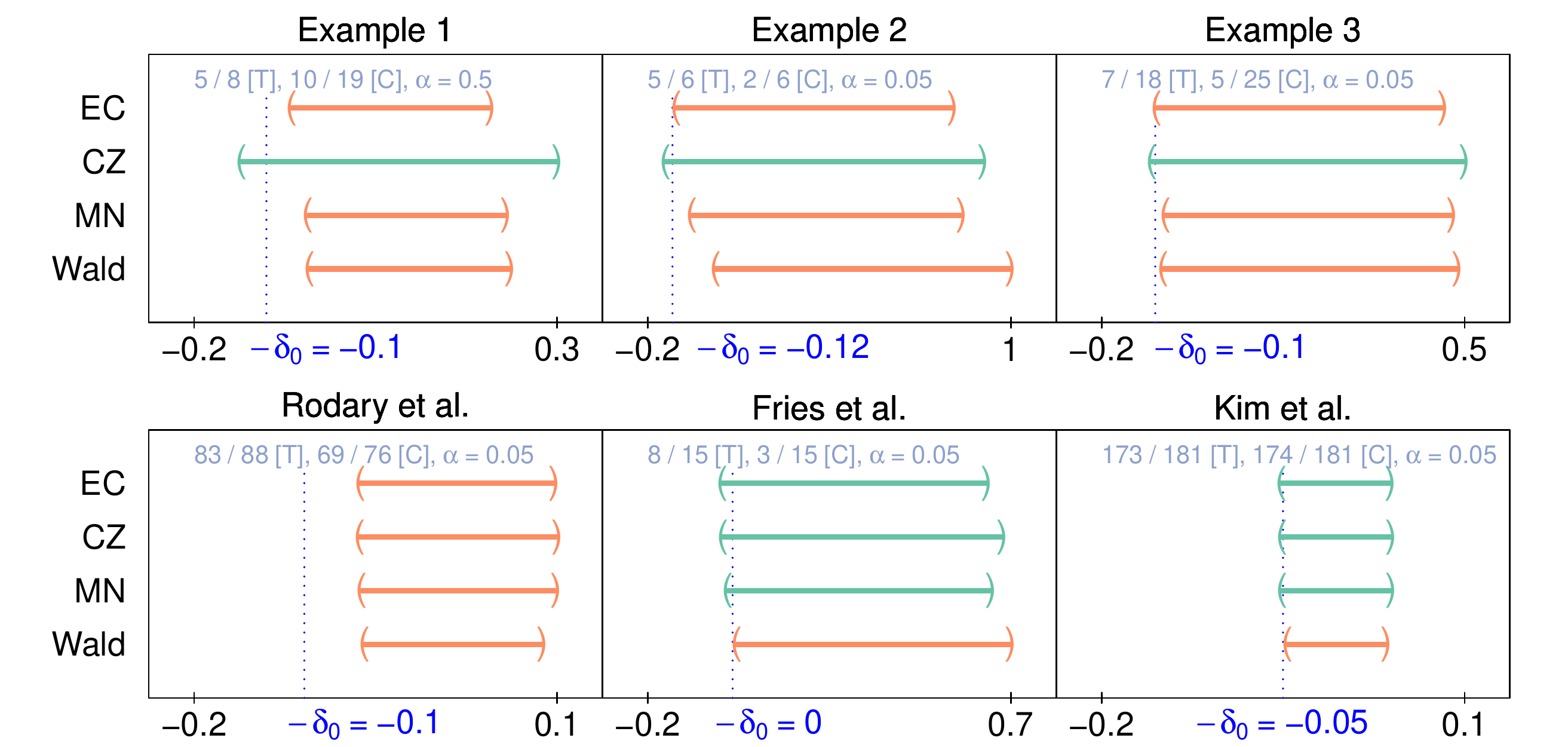}
    \caption{Comparison of the four confidence interval methods over the six data examples presented in Section \ref{sec:examples}.
    }
    \label{fig:ci}
\end{figure}

\section{Discussion}

A novel confidence interval estimator is proposed that bridges the divide between the generally more powerful asymptotic confidence interval of \cite{Miettinen:85} and the less powerful but correctly sized exact confidence interval of \cite{Chan:99} to yield a correctly sized exact confidence interval that is more powerful than the Chan \& Zhang interval. The proposed confidence interval fully leverages the noninferiority trial design by incorporating the noninferiority margin, where as the other methods do not involve the pre-specified noninferiority margin. 

For larger sample sizes the methods all produce similar confidence intervals, but the Chan \& Zhang confidence interval method is substantially more computationally demanding. Moderate sample sizes, such as the examples explored in Section \ref{sec:examples} can take from several minutes to several hours, depending on the level of precision required, whereas the other methods, including the proposed method, will compute the confidence interval within a few seconds.

The differences in results are usually not very dramatic, but with smaller sample sizes and certain values of the parameters $P_T$, $\alpha$, and $\delta_0$, the proposed method can provide a pretty substantial improvement in power as demonstrated in Section \ref{sec:power}. Theorems \ref{thm:chan-zhang} and \ref{thm:EC} also theoretically establish that the proposed exact-corrected confidence interval estimator is at least as powerful as the Chan \& Zhang confidence interval estimator and that they both correspond to valid p-value estimators with controlled size. Therefore, the proposed exact-corrected risk difference confidence interval estimator is recommended for noninferiority binomial trials as it is computationally efficient, preserves the type-I error, and has improved power over the Chan \& Zhang interval.


\bibliographystyle{agsm}
\bibliography{nour}

\end{document}